\documentclass[preprint,aps,prd,showpacs]{revtex4}
\usepackage{amsfonts}
\usepackage{mathrsfs}
\usepackage{amsmath}
\usepackage{graphicx}
\usepackage{color}

\usepackage{amsmath,amsthm,amssymb,amsfonts}
\usepackage{mathtools}
\usepackage{bm} % bold math font
\usepackage{ulem}
\usepackage{upgreek} % upright greek letters
\usepackage{pifont} % special symbol
\usepackage{blkarray} % bordered matrix
\usepackage[all]{xy} % draw diagram
\usepackage{enumitem} % extended list environments
\usepackage{subdepth} % align of subscript and superscript
\usepackage{fouridx} 
\usepackage{todonotes}
\usepackage{graphicx}
\usepackage[mode=image|tex]{standalone}
\usepackage{tikz}
\usepackage{calc}
\usetikzlibrary{calc}
\usetikzlibrary{positioning}
\usetikzlibrary{decorations.pathmorphing,decorations.markings}
\usetikzlibrary{arrows.meta}
\usetikzlibrary{backgrounds}
\usetikzlibrary{fit}

\usepackage{algorithm}
\usepackage{algpseudocode}

\usepackage[colorlinks,linkcolor=red,anchorcolor=blue,citecolor=green]{hyperref}
\usepackage[capitalise,noabbrev]{cleveref}

\newtheorem{theorem}{Theorem}

\newtheorem{corollary}[theorem]{Corollary}
\newtheorem{definition}[theorem]{Definition}
\newtheorem{example}[theorem]{Example}

\newtheorem{problem}[theorem]{Problem}
\newtheorem{remark}[theorem]{Remark}

\newtheorem{proposition}[theorem]{Proposition}

\newcommand{\bea}{\begin{eqnarray}}
\newcommand{\eea}{\end{eqnarray}}
\newcommand{\beq}{\begin{equation}}
\newcommand{\eeq}{\end{equation}}

\def\/{\over}
\def\be{\begin{equation}}
\def\ba{\begin{eqnarray}}
\def\ee{\end{equation}}
\def\ea{\end{eqnarray}}

\newcommand*\cplxconj[1]{\overline{#1}}

\newcommand{\injto}{\hookrightarrow}

\newcommand{\dd}{\,\mathrm{d}}

\DeclareMathOperator{\Tr}{Tr}

\DeclareMathOperator{\Hom}{Hom}

\DeclareMathOperator{\UG}{U}
\DeclareMathOperator{\Sp}{Sp}

\DeclareMathOperator{\SL}{SL}
\DeclareMathOperator{\Suz}{Suz}
\DeclareMathOperator{\PSL}{PSL}
\DeclareMathOperator{\MG}{M}
\DeclareMathOperator{\Rd}{Rd}
\DeclareMathOperator{\PSp}{PSp}
\DeclareMathOperator{\diag}{diag}

\DeclarePairedDelimiter{\set}{\{}{\}}
\DeclarePairedDelimiter{\abs}{\vert}{\vert}
\DeclarePairedDelimiter{\norm}{\Vert}{\Vert}

\DeclarePairedDelimiter{\lrbra}{[}{]}

\begin{document}
\title{\bf On the explicit constructions of certain unitary 
$t$-designs}
\author{Eiichi Bannai${}^{1}$~$\footnote{ bannai@math.kyushu-u.ac.jp}$}
\author{ Mikio Nakahara${}^{2,3}$~$\footnote{ nakahara@shu.edu.cn}$}
\author{ Da Zhao${}^{4}$~$\footnote{jasonzd@sjtu.edu.cn}$}
\author{ Yan Zhu${}^{2}$~$\footnote{ zhu\_yan@shu.edu.cn}$}
%\author{Shao-Feng Wu${}^{1}$~$\footnote{ sfwu@shu.edu.cn}$}
\affiliation{$^1$ Faculty of Mathematics, Kyushu University (emeritus), Japan\\
$^2$ Department of Mathematics, Shanghai University, Shanghai 200444, China\\
$^3$ Research Institute for Science and Technology, Kindai University, Higashi-Osaka, 577-8502, Japan\\
$^4$ School of Mathematical Sciences, Shanghai Jiao Tong University, Shanghai 200240, China
}

%\date{\today}

\vspace{3cm}

\begin{abstract}
	Unitary $t$-designs are ``good'' finite subsets of the unitary group $\UG(d)$ that approximate the whole unitary group $\UG(d)$ well. Unitary $t$-designs have been applied in
	randomized benchmarking, tomography, quantum cryptography and many other areas of
	quantum information science.
	If a unitary $t$-design itself is a group then it is called a unitary $t$-group. 
	Although it is known that unitary $t$-designs in $\UG(d)$ exist for any $t$ and $d$, the unitary $t$-groups do not exist for $t\geq 4$ if $d\geq 3$, as it is shown by Guralnick-Tiep (2005) and Bannai-Navarro-Rizo-Tiep (BNRT, 2018). 
	Explicit constructions of exact unitary $t$-designs in $\UG(d)$ are not easy in general. 
	In particular, explicit constructions of unitary $4$-designs in $\UG(4)$ have been an open problem in quantum information theory. 
	We prove that some exact unitary $(t+1)$-designs in the unitary group $\UG(d)$ are constructed from unitary $t$-groups in $\UG(d)$ that satisfy certain specific conditions. 
	Based on this result, we specifically construct exact unitary $3$-designs in $\UG(3)$ from the unitary $2$-group $\SL(3,2)$ in $\UG(3),$ and also unitary $4$-designs in $\UG(4)$ from the unitary $3$-group $\Sp(4,3)$ in $\UG(4)$ numerically. 
	We also discuss some related problems. 
\end{abstract}
\maketitle

\section{Introduction}

\baselineskip=16pt

The basic idea of ``design theory'' is to approximate a given space $M$ by a good finite subset $X$ of $M$. 
The spherical $t$-designs are those finite subsets $X$ of the unit sphere $M=S^{n-1}$ such that for any polynomial $f$ of degree up to $t$, the spherical integral of $f$ on the sphere is given by the average value of $f$ at the finitely many points of $X$ of $S^{n-1}$ \cite{MR0485471}. 
So are the concept of combinatorial $t$-designs ($t$-$(v,k,\lambda)$ designs) of the $M={\binom{V}{k}}$, the set of all the $k$-element subsets of a set $V$ of cardinality $v$. 
The space $M$ has the structure of an association scheme called Johnson association scheme $J(v,k)$. 
This concept of $t$-design was generalized further to the concept of $t$-designs in $Q$-polynomial association schemes by Delsarte \cite{MR0384310}. 
There are many different kinds of design theories and there is 
vast literature on various design theories. 
We would like to refer the readers, in particular, to the following two papers \cite{MR2535394,MR3594369} for the review of the developments of design theory including many generalizations of the concept of $t$-designs, from viewpoint of algebraic combinatorics. 

The microscopic world is described by quantum physics, where the time-evolution of a closed system is expressed by a unitary transformation. Accordingly, study of unitary transformations, or unitary matrices if the system is finite-dimensional, is essential to understand the quantum world. 
Needless to say, unitary transformations play central
roles in quantum computing and quantum information theory.
So, it is natural for us to approximate the whole unitary group $\UG(d)$ by a finite subset $X$ of $M=\UG(d)$. 
This lead physicists and mathematicians to formulate the concept of unitary $t$-designs \cite{MR2326329,MR2433437}. 
A systematic study of unitary $t$-designs from a mathematical viewpoint is given by Roy-Scott \cite{MR2529619} and we use their paper as a basic reference on unitary $t$-designs. 
There are many further developments on the theory of unitary $t$-designs, including those so called approximate unitary $t$-designs.
Those unitary $t$-designs which satisfy \cref{eqn:biaxal} in \cref{def:magnum} in Section II is called exact unitary $t$-designs. Approximate unitary $t$-designs have also been considered and studied mainly in physics.

A unitary $t$-design $X$ of $\UG(d)$ is called a unitary $t$-group
if $X$ is a subgroup of $\UG(d)$ as well.
In physics, cf \cite{Zhu_2017,Webb:2016:CGF:3179439.3179447,1609.08172}, some unitary 3-groups have been known, say Clifford groups and some sporadic examples, but the difficulty of finding unitary 4-groups (except for the case of $d=2$, cf. \cite{1810.02507}) has been noticed. 
Actually, the non-existence of unitary 4-groups was known for $d\geq 5$ in a disguised form in finite group theory, in a very deep paper of Guralnick-Tiep \cite{MR2123127}) that uses the classification of finite simple groups. 
This was recently pointed out by BNRT \cite{1810.02507} and the complete classification of unitary $t$-groups on $\UG(d)$ for all $t\geq 2$ and $d\geq 2$ was obtained
therein.
%in \cite{1810.02507}. 

Although unitary $4$-groups on $\UG(d)$ do not exist for $d\geq 3$ at all, unitary $t$-designs exist for all $t$ and $d$ as was proved in
%it is seen from the theorem of 
Seymour-Zaslavsky \cite{MR744857}. 
However, the explicit constructions of unitary $t$-designs are challenging in general, similarly as in the case for the explicit constructions of spherical $t$-designs.  In particular, 
while the existence of unitary $4$-designs in $\UG(4)$ have been known, their explicit constructions were not obtained so far to our knowledge \cite{Nakata}. 
Explicit constructions of unitary $t$-designs are essential in many areas of quantum information processing such as efficient randomized benchmarking of quantum channels
\cite{PhysRevA.77.012307,PhysRevLett.106.180504,PhysRevA.85.042311,wallman2014randomized,1510.02767,1711.08098}, 
quantum process tomography\cite{MR2433437,2016APS..MARB44006L}, 
quantum state tomography\cite{PhysRevA.72.032325,scott2006tight,PhysRevA.84.022327,PhysRevA.90.012115}, decoupling
\cite{szehr2013decoupling,nakata2017decoupling},
quantum cryptography\cite{ambainis2009nonmalleable}
 and data hiding\cite{985948}, among others. 
Their efficient implementation in terms of the number of local gates have been actively studied \cite{Cleve:2016:NCE:3179473.3179474,MR2551028,MR3535891,PhysRevX.7.021006}.

The main purpose of this paper is to give explicit constructions of 
unitary $3$-designs in $\UG(3)$ and 
unitary $4$-designs in $\UG(4)$ numerically. 
In order to do that, we first obtain the following purely mathematical theorem that explains how we can construct unitary $(t+1)$-designs from certain unitary $t$-group $G$ explicitly. 
Namely, we obtain the following Theorem:

\begin{theorem} \label{thm:Khazar}
	Let ${G}$ be a finite subgroup of $\UG(d)$, and let $\chi : \UG(d) \to \UG(d)$ be the natural (fundamental) unitary representation of $\UG(d)$. 
	We abuse the notation by considering $\chi : G \hookrightarrow \UG(d)$ as the natural embedding of $G$.
	Suppose that ${G}$ is a unitary $t$-group in ${\UG(d)}$.
	Let $\chi^{t+1}$ be the $(t+1)$ times tensor product of the fundamental representation $\chi$. 
	Suppose
	$${(\chi^{t+1} , \chi^{t+1})_G=(\chi^{t+1}, \chi^{t+1} )_{\UG(d)}+1.}$$ 
	Then there exists a non-zero  $G \times G$-invariant homogeneous polynomial $f \in \Hom (\UG(d), t+1, t+1)$, unique up to scalar multiplication, such that $\int_{\UG(d)} f(U) \dd U = 0$. 
	Let ${U_0\in \UG(d)}$ be a zero of ${f.}$ 
	Then the orbit ${X} = GU_0G$ of ${U_0}$ under the action ${G\times G}$ on ${\UG(d)}$ becomes a unitary $(t+1)$-design in ${\UG(d)}$.
\end{theorem}

Here we defined the inner product of two representations $\rho_1$ and $\rho_2$ of a group $\UG(d)$ by $(\rho_1 , \rho_2)_{\UG(d)} = \int_{\UG(d)} \Tr \rho_1(U) \overline{\Tr \rho_2(U)}
	 \dd U$ and $(\rho_1 , \rho_2)_{G} = \frac{1}{|G|} \sum_{x \in G} \Tr \rho_1(x) \overline{\Tr \rho_2(x)}$ for a finite subgroup $G \subset \UG(d)$. The Haar measure is normalized
	 as $\int_{\UG(d)} \dd U = 1$.

	This theorem guarantees that if there is such $G$ satisfying the conditions of \cref{thm:Khazar}, then there is a non-trivial homogeneous polynomial $f$ in $\Hom (\UG(d), t+1,t+1)$ that is invariant under the action of $G\times G.$  
	Take any zero $U_0$ of $f$ on $\UG(d)$, then the orbit of $U_0$ 
	under the action of $G\times G$, say $X=GU_0G$, gives a unitary $(t+1)$-design on $\UG(d).$ 
	In \cref{sec:Ino}, we apply this Theorem in particular for the two cases
	\begin{enumerate}
		\item $d=3, G=\SL(3,2), t=2$ 
		\item $d=4, G=\Sp(4,3), t=3$
	\end{enumerate}
	to construct the explicit unitary $(t+1)$-designs in $\UG(d)$ numerically. 

	This technique also works for other $G$ satisfying the conditions of \cref{thm:Khazar}, but the large order of the group so far prevented us from getting the explicit examples for other cases. 
	They should be manageable if we have more computational resources. 

	\cref{thm:Khazar} claims $X$ is a unitary $(t+1)$-design, although it does not rule out the possibility that $X$ is also a unitary $(t+2)$-design. 
	We have the following theorem to bound the strength of the design.

\begin{theorem} \label{thm:lentiform}
	Let $G$ be a finite subgroup of $\UG(d)$. 
	Let $X_1 = G U_1 G$ and $X_2 = G U_2 G$ be two orbits of the natural action $G \times G$ on $\UG(d)$. 
	Suppose $X_i$ is a unitary $t_i$-design but not a unitary $(t_i+1)$-design where $i = 1,2$. 
	Then $t_1 \leq 2 t_2 +1$ and $t_2 \leq 2 t_1 + 1$.
\end{theorem}

	This theorem is motivated by \cite{MR740321} which proves a similar result for spherical designs.

	We will conclude our paper by giving some discussions. 

\section{\texorpdfstring{Unitary $t$-designs and unitary $t$-groups}{Unitary t-designs and unitary t-groups}}

Let us recall the definition of unitary $t$-designs in $\UG(d).$ 

\begin{definition} \label{def:magnum}
	A finite subset ${X}$ of the unitary group ${\UG(d)}$ is called a unitary ${t}$-design, if 
	\begin{equation} \label{eqn:biaxal} 
		 {\int_{\UG(d)}f(U) \dd U=\frac{1}{|X|}\sum_{U\in X}f(U)}
	\end{equation}
	for any ${f(U)\in \Hom(\UG(d), t,t).}$
	Here ${\Hom (\UG(d),r,s)}$ is the space of polynomials that are homogeneous of degree ${r}$ in the matrix entries of ${U}$, and homogeneous of degree ${s}$ in the matrix entries of the Hermitian conjugate ${U^{\dagger}}$ of $U$.
\end{definition}

	Those satisfying the condition (\ref{eqn:biaxal}) above are called exact unitary $t$-designs in some literature. 
	In this paper, we consider only these unitary $t$-designs. 
	While those with the condition (\ref{eqn:biaxal}) replaced by the condition that the difference of both sides is very small, are called approximate unitary $t$-designs. 
	Of course exact unitary $t$-designs are approximate unitary $t$-designs, and both types of unitary $t$-designs are studied extensively in physics \cite{MR2551028,MR3535891,PhysRevX.7.021006,dankert2009exact,harrow2009random}.
	
	It is known that there are many equivalent characterizations of unitary $t$-design in $\UG(d)$. (cf. Roy-Scott \cite{MR2529619}, Zhu-Kueng-Grassl-Gross 
	\cite{1609.08172}.) 
	Here, we will use some of the equivalent conditions later in our paper.  
	One equivalent definition is as follows \cite[p.14]{MR2529619}:

	A finite subset $X$ in $\UG(d)$ is a unitary $t$-design, if and only if for any $f\in \Hom (\UG(d), t,t)$, 
	$$(1,f)_{\UG(d)}=(1,f)_X,$$ 
	where $$(1,f)_{\UG(d)}=\int_{\UG(d)}\overline{f}\dd U$$ and 
	$$(1,f)_{X}=\frac{1}{|X|}\sum_{x\in X}\overline{f(x)}.$$ \\
%We take the total Haar measure on $\UG(d)$, so $\abs{\UG(d)}$ is normalized to $1$.
	
	There are several different characterization of unitary $t$-groups \cite[Corollary 8]{MR2529619} and \cite[Proposition 3]{1609.08172}. 
	Let ${\chi}$ be the natural (fundamental) representation of ${G \hookrightarrow \UG(d)}$ as well as the natural representation of $\UG(d)$,  $\chi : \UG(d) \hookrightarrow \UG(d)$. 
	The notation $\chi^t$ is the shorthand for $t$ times tensor product $\chi \otimes \cdots \otimes \chi$. 
	\begin{enumerate}
		\item A finite subgroup ${G}$ is a unitary ${t}$-group in ${\UG(d)}$, if and only if
	$${\frac{1}{|G|}\sum_{g\in G}\abs{\Tr \chi(g)}^{2t}=\int_{U\in \UG(d)}
	\abs{\Tr \chi(U)}^{2t} \dd U}.$$
		\item A finite subgroup $G$ is a unitary $t$-group if and only if the decomposition of $\UG(d)^{\otimes t}$ into the irreducible representations of ${\UG(d)}$ is the same as the decomposition of $G^{\otimes t}$ into the irreducible representations of ${G}$ in the sense of both dimension and multiplicity.  
		\item A finite subgroup ${G\subset \UG(d)}$ is a unitary ${t}$-group, if and only if 
	$${M_{2t}(G, V)=M_{2t}(\UG(d),V).}$$ 
	where the LHS 
	$${M_{2t}(G, V) :=(\chi^t, \chi^t)_G=\frac{1}{|G|}\sum_{g\in G}{\chi^t}(g)\overline{{\chi^t}(g)} = \frac{1}{\abs{G}}\sum_{g\in G}\abs{\Tr \chi(g)}^{2t}}$$ 
	and the RHS ${M_{2t}(\UG(d),V)}$ is the corresponding inner product 
	$${M_{2t}(\UG(d),V) := (\chi^t, \chi^t)_{\UG(d)}=\int_{U\in \UG(d)} \abs{\Tr \chi(U)}^{2t} \dd U.}$$
	\end{enumerate}

	Let us recall that unitary $t$-groups in $\UG(d)$ are completely classified for all $t\geq 2$ and $d\geq 2.$ (Cf. Guralnick-Tiep \cite{MR2123127} and BNRT \cite{1810.02507}.) 
	The main %result 
	purpose of this paper is to prove \cref{thm:Khazar} given in Section~II and construct new unitary designs accordingly.

\section{Proofs}  

	It is known that the irreducible representations of $\UG(d)$ appearing in $\chi^{t+1} \otimes \overline{\chi}^{t+1}$ are parametrized by the non-increasing integer sequence $\mu =(\mu_1, \mu_2, \ldots , \mu_d).$ 
	The irreducible representation of $\UG(d)$ corresponding the sequence $\mu$ is denoted by $\rho_{\mu}.$ (cf. \cite{MR2529619}). 
	Let $\Phi $ be the set of $\mu$ with $\rho_{\mu}$ in the representation $\chi^{t+1} \otimes \overline{\chi}^{t+1}$. 
	Such $\mu$ is characterized by $\mu_+ = - \mu_- \leq t+1$. 
	Here, $\mu_+$ is the sum of all positive $\mu_i$'s  and $\mu_-$ is the sum of all negative $\mu_i$'s. 
	% And we also assume that $d\geq t+1$ unless otherwise stated. 

	Let $G$ be a subgroup of $\UG(d)$, and let $\chi$ be the natural embedding of $G$ into $\UG(d)$. 
	Suppose that $d\geq t+1$. 
	Let $(\chi^{t+1} \overline{\chi}^{t+1}, 1)_G = (t+1)! +1$. 
	First we prove the following proposition:

\begin{proposition} \label{prop:Kalmarian}
	With the notation given above, there is a unique non-trivial irreducible representation $\rho_{\widetilde{\mu}}$ such that $(\rho_{\widetilde{\mu}}, 1)_G=1$, where $\widetilde{\mu} \in \Phi$.
\end{proposition}

\begin{proof}
	We write $\displaystyle \chi^{t+1} =\bigoplus_\lambda H_{\lambda}=\bigoplus_\lambda W_{\lambda}\otimes S_{\lambda}$ as given in \cite[p.12]{1609.08172}. Here $W_{\lambda}$ is the Weyl module carrying the irreducible representation of $\UG(d)$ associated with the partition $\lambda$ while $S_{\lambda}$ is the Specht module of which
	the symmetric group $S_{t+1}$ acts irreducibly. 
	
By our assumption $(\chi^{t+1} \overline{\chi}^{t+1},1)_G=\sum_{\lambda, \rho} d_{\lambda} d_{\tau}(W_{\lambda}
	\cplxconj{W_{\tau}}, 1)_G =(t+1)!+1$. Here $\lambda$ and $\tau$ are non-increasing partitions of $t+1$ into no more than $d$ parts and $d_{\lambda}$ is the degree of the Specht module $S_\lambda$. 
	Note that 
	$(\chi^{t+1} \overline{\chi}^{t+1},1)_{\UG(d)}=\sum_{\lambda, \tau} d_{\lambda} d_{\tau} (W_{\lambda}
	\cplxconj{W_{\tau}}, 1)_{\UG(d)} = (t+1)! =\sum{d_{\lambda}}^2$. 
	
On the other hand, the irreducible representations $\rho_{\mu}$ of $\UG(d)$ appearing in $\chi^{t+1} \overline{\chi}^{t+1}$ are characterized in \cite[Theorem 4]{MR2529619}.  
We do not know the exact multiplicities in which each irreducible representation $\rho_{\mu}$ (of $\UG(d)$) appearing in $\chi^{t+1} \overline{\chi}^{t+1}$. 
However, we know that the multiplicity of $\rho_{(0,...,0)}$ is $(t+1)!$.   
	Since $(\rho_{\mu},1)_{\UG(d)}= 0$ for $\mu\not={(0,...,0)}$, we conclude that there is exactly one $\widetilde{\mu} \neq {(0,...,0)}$ such that $(\rho_{\tilde{\mu}},1)_G=1$. 
\end{proof}

Next we introduce the concept of unitary $\rho$-design for later proof.

\begin{definition}
	Let $(\rho,V)$ be a unitary representation of $\UG(d)$. 
	Let $X$ be a finite subset of $\UG(d)$. 
	Then $X$ is called a unitary $\rho$-design if 
	\begin{equation} \label{eqn:unshavenness}
		\frac{1}{\abs{X}} \sum_{U \in X} \rho(U) = \int_{\UG(d)} \rho(U) \,\mathrm{d}U.
	\end{equation}
\end{definition}

Obviously we have another characterization of unitary $t$-design.

\begin{theorem} \label{thm:walkmiller}
	$X$ is a unitary $t$-design if and only if $X$ is a unitary $\rho$-design for every irreducible representation $\rho$ appearing in $U^{\otimes t} \otimes (U^\dagger)^{\otimes t}$.
\end{theorem}

We mimic the proof in \cite[Theorem 5.4]{MR2433437} to get an equivalent definition of unitary $\rho$-design, whose condition is easy to confirm.

\begin{theorem} \label{thm:consortion}
	For any finite $X \subset \UG(d)$, 
	\begin{equation} \label{eqn:fillemot}
		\frac{1}{\abs{X}^2} \sum_{U,V \in X} \Tr\rho(U^\dagger V) \geq \int_{\UG(d)} \Tr\rho(U) \,\mathrm{d}U.
	\end{equation}
	with equality if and only if $X$ is a unitary $\rho$-design.
\end{theorem}

\begin{corollary} \label{coro:benevolent}
	If $X$ is a unitary $\rho$-design, then $UX$ is also a unitary $\rho$-design for every $U \in \UG(d)$.
\end{corollary}

\begin{proof}[Proof of \cref{thm:consortion}]
	Let $S := \frac{1}{\abs{X}} \sum_{U \in X} \rho(U) - \int_{\UG(d)} \rho(U) \,\mathrm{d}U$. 
	Then
	\begin{align*}
		0 \leq \Tr(S^\dagger S) &= \frac{1}{\abs{X}^2} \sum_{U,V \in X} \Tr\rho(U^\dagger V) - 2 \times \frac{1}{\abs{X}} \sum_{U \in X} \int_{\UG(d)} \Tr\rho(U^\dagger V) \,\mathrm{d}V \\
		&+ \int_{\UG(d)} \int_{\UG(d)} \Tr\rho(U^\dagger V) \,\mathrm{d}U \,\mathrm{d}V \\
		&= \frac{1}{\abs{X}^2} \sum_{U,V \in X} \Tr\rho(U^\dagger V) - \int_{\UG(d)} \Tr\rho(U) \,\mathrm{d}U. \qedhere
	\end{align*}
\end{proof}

Now we are able to prove \cref{thm:Khazar}.

\begin{proof}[Proof of \cref{thm:Khazar}]
	First we recall the fact that all the matrix coefficient functions of $\rho_\mu$ where $\mu \in \Phi$ form a basis of $\Hom(\UG(d), t+1,t+1)$ \cite[Theorem 7]{MR2529619}. 
	By \cref{prop:Kalmarian}, there is a unique non-trivial irreducible representation $\rho_{\widetilde{\mu}}$ such that $(\rho_{\widetilde{\mu}},1)_G = 1$. 
	By symmetry $\widetilde{\mu} = - \widetilde{\mu}$. Here $- \widetilde{\mu}$ is obtained by negating entries of $\widetilde{\mu}$ and put them in reverse order.

	For every non-trivial $\mu \neq \widetilde{\mu}$, we have shown that $(\rho_\mu,1)_{G} = 0$. 
	Therefore $\frac{1}{\abs{G}^2} \sum_{U,V \in G} \Tr\rho_\mu(U^\dagger V) = 0$. 
	By \cref{thm:consortion}, $G$ is a unitary $\rho_\mu$-design. 
	It implies that $\frac{1}{\abs{G}} \sum_{U \in G} \rho_\mu(U) = 0$. 
	Hence every matrix coefficient function of $\rho_\mu$ becomes $0$ after $G \times G$ averaging.

	For $\mu = \widetilde{\mu}$, let us consider its matrix coefficient functions. 
	Let $A = \frac{1}{\abs{G}} \sum_{U \in G} \rho_{\widetilde{\mu}}(U)$ and $M = M(U) = \frac{1}{\abs{G}^2} \sum_{g_1,g_2 \in G} \rho_{\widetilde{\mu}}(g_1^\dagger U g_2)$. 
	Note that $\rho_{\widetilde{\mu}}(g) M = M \rho_{\widetilde{\mu}}(g) =M$ for every $g \in G$ and hence $AM=MA=M$. 
	Suppose $\rho_{\widetilde{\mu}}$ decomposes into irreducible representations $(\rho_\eta, V_\eta)$ of $G$ by $\rho_{\widetilde{\mu}} = \bigoplus_{\eta \in \Gamma} m_\eta \rho_\eta$. 
	Then $M$ is a block diagonal matrix with blocks corresponding to these $\rho_\eta$.
	By \cref{prop:Kalmarian}, one of the $\rho_\eta$'s is the trivial representation and its multiplicity is one. 
	For every other $\eta$, since $(\rho_\eta,1)_G = 0$, we have $A|_{V_\eta} = 0_{V_\eta}$. 
	Therefore $M|_{V_\eta} = 0_{V_\eta}$. 
	Hence the matrix coefficient functions in the block corresponding to $\rho_\eta$ becomes $0$ as well after $G \times G$ averaging.
	Note that the trivial representation in $\rho_{\widetilde{\mu}}$ is of dimension 1 and of multiplicity 1, therefore besides the trivial constant function only one matrix coefficient is non-zero after $G \times G$ averaging. 
	In fact, its $G \times G$ averaging is equal to the polynomial $f(U) = \frac{1}{\abs{G}^2} \sum_{g_1, g_2 \in G} \chi_{\widetilde{\mu}}(g_1^\dagger U g_2)$.
	Note that $(\rho_{\widetilde{\mu}},1)_{\UG(d)} = 0$ implies that $\int_{\UG(d)} f(U) \dd U = 0$.
	So far, we have shown the existence and uniqueness of the non-zero  $G \times G$-invariant homogeneous polynomial $f \in \Hom (\UG(d), t+1, t+1)$ such that $\int_{\UG(d)} f(U) \dd U = 0$.

	Now we take a zero $U_0$ of the polynomial $f(U) = \frac{1}{\abs{G}^2} \sum_{g_1, g_2 \in G} \chi_{\widetilde{\mu}}(g_1^\dagger Ug_2)$. 
	Let $X$ be the orbit of $U_0$ under the action of $G \times G$ on $\UG(d)$. 
	For every non-trivial $\mu \neq \widetilde{\mu}$, we have shown that $G$ is a unitary $\rho_\mu$-design. 
	By \cref{coro:benevolent} and the additivity of unitary $\rho$-design, $X = GU_0G$ is a unitary $\rho_\mu$-design. 
	For $\mu = \widetilde{\mu}$, since $U_0$ is a zero of $f$, we get $\Tr M(U_0) = 0$. 
	Combined with the argument in the last paragraph, $M(U_0)$ is indeed the zero matrix. 
	Hence $X$ is a unitary $\rho_{\widetilde{\mu}}$-design.

	Finally by \cref{thm:walkmiller}, we conclude that $X$ is a unitary $(t+1)$-design.
\end{proof}

\begin{proof}[Proof of \cref{thm:lentiform}]
	Without loss of generality, let us assume that $t_1 \leq t_2$. 
	Since $X_1 = G U_1 G$ is not a unitary $(t_1 + 1)$-design, there exists a $G \times G$-invariant homogeneous polynomial $h \in \Hom(\UG(d),t_1+1,t_1+1)$ such that $h(U_1) \neq 0$ and $\int_{\UG(d)} h(U) \dd U = 0$. 
	Now let us consider the $G \times G$-invariant homogeneous polynomial $h\overline{h} \in \Hom(\UG(d), 2t_1 + 2, 2t_1 +2)$. 
	Note that $c := \int_{\UG(d)} (h\overline{h})(U) \dd U > 0$. 
	Let $f := h \overline{h} - c \in \Hom(\UG(d), 2t_1 + 2, 2t_1 +2)$, then $\int_{\UG(d)} f(U) \dd U = 0$. 
	Suppose $X_2$ is a unitary $(2t_1 +2)$-design, then we must have $f(U_2) = 0$. 
	Note that $2t_1 + 2 > t_1 +1$, so $h(U_2) = 0$.
	Therefore $f(U_2) = h(U_2)\overline{h}(U_2) - c = -c < 0$, contradiction. 
	Hence $t_2 \leq 2t_1 + 1$.
\end{proof}

\section{\texorpdfstring{Examples of unitary $t$-groups $G$ in $\UG(d)$ satisfying the conditions of Theorem 1}{Examples of unitary t-groups G in U(d) satisfying the conditions of Theorem 1}}

The followings are some examples of ${G\subset \UG(d)}$ that satisfy the conditions in \cref{thm:Khazar}. Here, we basically use the notation of An Atlas of Finite Groups \cite{MR827219}. Also, see Guralnick-Tiep \cite{MR2123127} 
and BNRT \cite{1810.02507}. 

\begin{enumerate}
	\item For ${t=3,}$ (We assume ${d\geq 3.}$)
	\begin{enumerate}
		\item ${d=4, G=\Sp(4,3),}$
		\item ${d=6, G=6_1.\UG_4(3),}$
		\item ${d=12, G=6\Suz .}$
	\end{enumerate}
	\item For ${t=2,}$ (We assume ${d\geq 3.}$)
	\begin{enumerate}
		\item ${d=3, G=\SL(3,2)=\PSL(2,7),}$
		\item ${d=10, G=\MG_{22},}$
		\item ${d=28, G=\Rd ,}$
		\item ${d=(3^m\pm 1)/2, G=\PSp(2m,3), \Sp(2m,3)}$. 
		(see \cite[Section 4]{1810.02507} for the details of Weil representations in this case.) 
	\end{enumerate}
\end{enumerate}

The above list might exhaust all such examples, although we will not try to give a rigorous proof of this claim.

\section{Computation} \label{sec:Ino}

\subsection{The unitary representation of \texorpdfstring{$\SL(3,2)$}{SL(3,2)} and \texorpdfstring{$\Sp(4,3)$}{Sp(4,3)}}
	We aim to construct some unitary $(t+1)$-designs based on certain unitary $t$-groups. 
	This urges us to find the unitary representations of these groups first.

	We adopt the notation $E(n)$ being the $n$-th root of unity from the mathematical software GAP \cite{GAP4}. 
	The following two constructions are taken from \cite[Equation 10.1 and Equation 10.5]{MR0059914}.

	\begin{example} \label{eg:floatboard}
		Let $a := -(E(7)^4+E(7)^2+E(7))$.
		Let $\mathcal{M}$ be the matrix group generated by the following three matrices. 
		$$
		M_1 = \begin{bmatrix}
		 1 &  &  \\
		  &  & 1 \\
		  & 1 & 
		\end{bmatrix},
		\ M_2 = \begin{bmatrix}
		 1 &  &  \\
		  & 1 &  \\
		  &  & -1
		\end{bmatrix},
		\ M_3 = \begin{bmatrix}
		 1/2 & -1/2 & -a/2 \\
		 -1/2 & 1/2 & -a/2 \\
		 -\cplxconj{a}/2 & -\cplxconj{a}/2 & 0
		\end{bmatrix}.
		$$
		Then $G=\mathcal{M}^{(1)}$, the commutator subgroup of $\mathcal{M}$, is isomorphic to $\SL(3,2)$ and is embedded in $\UG(3)$. 
	\end{example}
	
	\begin{example} \label{eg:zoon}
		Let $\omega := E(3)$. 
		Let $\mathcal{M}$ be the matrix group generated by the following four matrices. 
		\begin{align*}
		M_1 = \begin{bmatrix}
		 1 &  &  &  \\
		  & 1 &  &  \\
		  &  & \omega^2 &  \\
		  &  &  & 1
		\end{bmatrix}, & 
		\ M_2 = \frac{-i}{\sqrt{3}} \begin{bmatrix}
		 \omega^{\phantom{1}} & \omega^2 & \omega^2 & 0 \\
		 \omega^2 & \omega^{\phantom{1}} & \omega^2 & 0 \\
		 \omega^2 & \omega^2 & \omega^{\phantom{1}} & 0 \\
		 0 & 0 & 0 & i\sqrt{3}
		\end{bmatrix}, \\
		M_3 = \begin{bmatrix}
		 1 &  &  &  \\
		  & \omega^2 &  &  \\
		  &  & 1 &  \\
		  &  &  & 1
		\end{bmatrix}, & 
		\ M_4 = \frac{-i}{\sqrt{3}} \begin{bmatrix}
		 \phantom{-}\omega^{\phantom{1}} & -\omega^2 & 0 & -\omega^2 \\
		 -\omega^2 & \phantom{-}\omega^{\phantom{1}} & 0 & \phantom{-}\omega^2 \\
		 0 & 0 & i\sqrt{3} & 0 \\
		 -\omega^2 & \phantom{-}\omega^2 & 0 & \phantom{-}\omega^{\phantom{1}} 
		\end{bmatrix}.
		\end{align*}
		Then $G=\mathcal{M}^{(1)}$, the commutator subgroup of $\mathcal{M}$, is isomorphic to $\Sp(4,3)$ and is embedded in $\UG(4)$.
	\end{example}

\subsection{The \texorpdfstring{$G \times G$}{GxG}-invariant polynomial}
	The construction of the $G \times G$-invariant polynomial $f$ in $\Hom(\UG(d),t+1,t+1)$ is based on the irreducible characters of $\UG(d)$. 

	Suppose $\chi_\mu$ is the character of an irreducible representation $(\rho_\mu, V_\mu)$ of the unitary group $\UG(d)$. 
	It naturally induces a $G \times G$-invariant function on $\UG(d)$, namely
	\begin{equation} \label{eqn:trifledom}
		f(U) = \sum_{(g_1,g_2) \in G \times G} \chi_\mu (g_1^\dagger U g_2).
	\end{equation}

	A closed form of $\chi_\mu(\Lambda)$ can be expressed as a symmetric polynomial with respect to the spectrum of the unitary matrix $\Lambda$. 
	Note that if $\mu = -\mu$, then $\chi_\mu$, thus $f$, is a real function.

	\begin{theorem}[{\cite[Theorem 8]{MR2529619} or \cite[Theorem 38.2 and Proposition 38.2]{MR2062813}}] \label{thm:inhibitor}
		Let $V_\mu$ be the irreducible representation of $\UG(d)$ indexed by non-increasing integer sequence $\mu$. 

		If $\mu_d = 0$, then the character of $V_\mu$ is
		\begin{equation*}
			\chi_\mu (\Lambda) = s_\mu (\lambda_1, \ldots, \lambda_d),
		\end{equation*}
		where $s_\mu$ is the Schur polynomial, and $\set{\lambda_1, \ldots, \lambda_d}$ are the eigenvalues of $\Lambda$. 

		If $\mu_d \neq 0$, then the character of $V_\mu$ is
		\begin{equation*}
			\chi_\mu (\Lambda) = \det(\Lambda)^{\mu_d} \chi_{\mu'}(\Lambda),
		\end{equation*}
		where $\mu' = (\mu_1 - \mu_d, \ldots, \mu_{d-1} - \mu_d, 0)$.
	\end{theorem}

	For numerical computation, it takes considerable time to find the eigenvalues of a matrix and meantime it loses accuracy. 
	Therefore we prefer to express $\chi_\mu$ by $\Tr(\Lambda^k)$ and $\cplxconj{\Tr({\Lambda}^k)}$ where $0 \leq k \leq d$.
	This can be done by Newton--Girard formulae \cite[\S 10.12, pp. 278-279]{MR1740388}.

	\begin{example}
		By \cref{thm:inhibitor}, we have
		\begin{align*}
			& \chi_{(3,0,-3)}(\Lambda)  = \det(\Lambda)^{-3} s_{(6,3,0)}(\lambda_1,\lambda_2,\lambda_3) \\
			& = \frac{1}{(\lambda_1 \lambda_2 \lambda_3)^3} \bigg(
			-2 \lambda_1 \lambda_2 \lambda_3 \left(\lambda_1 \lambda_2+\lambda_3 \lambda_2+\lambda_1 \lambda_3\right) \left(\lambda_1+\lambda_2+\lambda_3\right){}^4 \\
			& +\left(\lambda_1 \lambda_2+\lambda_3 \lambda_2+\lambda_1 \lambda_3\right){}^3 \left(\lambda_1+\lambda_2+\lambda_3\right){}^3+2 \lambda_1^2 \lambda_2^2 \lambda_3^2 \left(\lambda_1+\lambda_2+\lambda_3\right){}^3 \\
			& +3 \lambda_1 \lambda_2 \lambda_3 \left(\lambda_1 \lambda_2+\lambda_3 \lambda_2+\lambda_1 \lambda_3\right){}^2 \left(\lambda_1+\lambda_2+\lambda_3\right){}^2-2 \left(\lambda_1 \lambda_2+\lambda_3 \lambda_2+\lambda_1 \lambda_3\right){}^4 \left(\lambda_1+\lambda_2+\lambda_3\right) \\ 
			& -5 \lambda_1^2 \lambda_2^2 \lambda_3^2 \left(\lambda_1 \lambda_2+\lambda_3 \lambda_2+\lambda_1 \lambda_3\right) \left(\lambda_1+\lambda_2+\lambda_3\right)+\lambda_1^3 \lambda_2^3 \lambda_3^3+2 \lambda_1 \lambda_2 \lambda_3 \left(\lambda_1 \lambda_2+\lambda_3 \lambda_2+\lambda_1 \lambda_3\right){}^3
			\bigg)
		\end{align*}
		Note that $\frac{1}{\lambda_i} = \cplxconj{\lambda_i}$ and $\Tr(\Lambda^k) = \sum_i {\lambda_i^k}$. We can simplify the above expression by Newton-Girard formulae of symmetric polynomials.
		\begin{align} \label{eqn:amphitriaene}
			&\chi_{(3,0,-3)}(\Lambda) = 
				\Tr(\Lambda^2)\cplxconj{\Tr(\Lambda)}^2
				+ \Tr(\Lambda^3)\cplxconj{\Tr(\Lambda) \Tr(\Lambda^2)} 
				+ \Tr(\Lambda)^2 \cplxconj{\Tr(\Lambda^2)}\\ 
				&- 2 \Tr(\Lambda^2) \cplxconj{\Tr(\Lambda^2)} 
				+ \Tr(\Lambda) \Tr(\Lambda^2) \cplxconj{\Tr(\Lambda^3)} 
				- \Tr(\Lambda^3) \cplxconj{\Tr(\Lambda^3)} 
				- 3 \Tr(\Lambda) \cplxconj{\Tr(\Lambda)} 
				+ 10 \nonumber
		\end{align}
	\end{example}

	\begin{example}
		\begin{align} \label{eqn:northwestern}
			&\chi_{(4,0,0,-4)}(\Lambda) = \bigg(
			\abs*{18 \Tr(\Lambda^4) - 12 \Tr(\Lambda) \Tr(\Lambda^3) - 6 \Tr(\Lambda^2)^2 + 4 \Tr(\Lambda^2) \Tr(\Lambda)^2}^2 \nonumber \\ 
			&+ 48 \Re \lrbra*{(2 \Tr(\Lambda) \Tr(\Lambda^3) + \Tr(\Lambda^2)^2) \cplxconj{\Tr(\Lambda^2) \Tr(\Lambda)^2}} \nonumber \\
			&- 16 \abs*{\Tr(\Lambda^2) \Tr(\Lambda)^2}^2 + \abs*{24 \Tr(\Lambda^3) - 27 \Tr(\Lambda^2) \Tr(\Lambda) + 3 \Tr(\Lambda)^3}^2 \\ 
			&- \abs*{3 \Tr(\Lambda)^3 - 27 \Tr(\Lambda) \Tr(\Lambda^2)}^2 + 360 \abs*{\Tr(\Lambda^2)}^2 + 216 \abs*{\Tr(\Lambda)^2}^2 \nonumber \\
			&- 1296 \Re \lrbra*{\cplxconj{\Tr(\Lambda^2)} \Tr(\Lambda)^2} + 432 \abs*{\Tr(\Lambda) \Tr(\Lambda^2)}^2 - 720 \abs*{\Tr(\Lambda)}^2 - 5040
			\bigg) / 144 \nonumber
		\end{align}
	\end{example}

\subsection{The approximation algorithm}

	Now our goal is reduced to the following problem.

	\begin{problem}
		Given a continuous real function $f$ defined on a connected Lie group, find a zero of this function (numerically).
	\end{problem}

	In particular, the function $f$ is a non-trivial $G \times G$-invariant polynomial on a unitary group $\UG(d)$. 
	The unitary group $\UG(d)$ is connected, and the existence of zero is guaranteed because the integration of $f$ on $\UG(d)$ is $0$. 

	Suppose $f(L) < 0$ and $f(R) > 0$ where $L, R$ are two matrices representing the elements of the Lie group. 
	By intermediate value theorem, there exists at least one matrix $Z$ on a path connecting $L$ and $R$ such that $f(Z) = 0$. 
	There are infinitely such paths and we will choose some special paths in the following.

	It is natural to use bisection method or false position method to approximate the zero in arbitrary precision. 
	The trouble here is that the function is defined on a manifold rather than the Euclidean space. 
	For Lie groups, there is a canonical atlas given by the exponential map from the Lie algebra to the Lie group. 
	We take advantage of this property to define the mid-point and the false position. 
	The mid-point of $L$ and $R$ is defined to be $\exp \left((\log L + \log R)/2 \right)$, and the false position between $L$ and $R$ is defined to be $\exp \left(\frac{f(R)\log L - f(L)\log R}{f(R)-f(L)} \right)$. 
	The false position method usually converges faster than the bisection method. 
	Nevertheless we use the bisection method when $L$ and $R$ are far away for the sake of robustness.
	One may consider other iterative methods to speed up the convergence. We did not use them because evaluation of the function is the heavy part of the computation. 
	The initial value of $L$ and $R$ are obtained by taking unitary matrices randomly until both of them are found. 

	\begin{algorithm}
	    \caption{FindZeroOnUnitaryGroup}\label{alg:Pocket}
	    \begin{algorithmic}[1]
	        \Function{FindZeroOnUnitaryGroup}{$f,L,R,\epsilon$}
	            \While{$\norm{L-R} > \epsilon$}
	            	\State{$M \gets \exp \left((\log L + \log R)/2 \right)$ or $\exp \left(\frac{f(R)\log L - f(L)\log R}{f(R)-f(L)} \right)$.}
	            \If{$f(M) < 0$}
	            	\State{$L \gets M$}
	            \Else
	            	\State{$R \gets M$}
	            \EndIf
	            \EndWhile 
	            \State{\Return{$M$}}
	        \EndFunction
	    \end{algorithmic}
	\end{algorithm}

	We are ready to construct the unitary designs, but let us put further constraint on the solution for the moment. 

	\begin{problem} \label{prob:septomarginal}
		Find a zero $U_0$ with good property, namely the size of the orbit $GU_0G$ is as small as possible.
	\end{problem}

	Suppose $GU_0G$ is an orbit whose size is smaller than $|G|^2$, then there must exist $g_1,g_2,g_3,g_4 \in G$, such that $g_1^\dagger U_0g_2=g_3^\dagger U_0g_4$. 
	Therefore $U_0^\dagger g_i U_0 = g_j$ where $g_i = g_3 g_1^\dagger $ and $g_j = g_4 g_2^\dagger $ are also elements of $G$. 
	This implies that $g_i$ and $g_j$ have the same spectrum. 
	If $g_i$ has distinct eigenvalues, then $U_0$ is on a submanifold isomorphic to $\UG(1) \times \UG(1) \times \cdots \times \UG(1)$. 
	If the eigenvalues of $g_i$ are not simple, then $U_0$ is on a submanifold isomorphic to $\UG(m_1) \times \UG(m_2) \times \cdots \times \UG(m_k)$, where $m_1, m_2, \ldots, m_k$ are the multiplicities of the eigenvalues.
	Note that there is no guarantee that a zero exists on the submanifold.

	Though it does not solve \cref{prob:septomarginal} completely, we have the clue to find them. 

\subsection{Solutions}

	For $G \cong \SL(3,2) \injto \UG(3)$, we find a zero on the submanifold $\UG(1) \times \UG(1) \times \UG(1)$, namely the diagonal unitary matrices. The size of the orbit is at most $\abs{G}^2/4 = 7056$. 

	\begin{example} \label{eg:cattlegate}
		Let $G \cong \SL(3,2)$ be the matrix group in \cref{eg:floatboard}, and let $f$ be the $G \times G$-invariant polynomial induced by the irreducible character $\chi_{(3,0,-3)}$ in \cref{eqn:amphitriaene}. 
		Then
		$U_0 = \diag(u_{11}, u_{22},u_{33})$
		is a zero of $f$, where $u_{11} = 1$, $u_{22} \approx 0.6480674529649858 - 0.7615829412529393 i$ and $u_{33} \approx -0.3307476956662597 - 0.9437192176762438 i$. 
		The error bound in \cref{alg:Pocket} is $\epsilon = 10^{-15}$.
		Hence the orbit $GU_0G$ is a unitary $3$-design on $\UG(3)$.
		The size of this orbit is $7056$.

	Moreover, we can characterize all the diagonal unitary matrices in $\UG(3)$ which make $G U_0 G$ a unitary $3$-design. 
	Let $u, v, t$ be real numbers and let $U_0 = \diag \left(e^{it}, e^{i(t + \frac{u+v}{2})}, e^{i(t + \frac{u-v}{2})}\right)$. 
	Then $G U_0 G$ is a unitary $3$-design if and only if $u$ and $v$ satisfy
	\begin{align} \label{eqn:brachiocyllosis}
		& \cos (u) [281838 \cos (v) - 156 \cos (2 v)-158]+ \sqrt{7} \sin (u) \left[24  \cos (v)+6  \cos (2 v)+2 \right] \nonumber\\
		& + [28125 -181 \cos (v)+140901 \cos (2 v)-65 \cos (3 v)] = 0.
	\end{align}
	%\cref{eqn:brachiocyllosis}. 
The solution of this equation is shown in \cref{fig:Beethovenish}.

\begin{figure}[htbp]
\centering
%\includestandalone{curve_uv.eps}
\includegraphics[scale=0.7]{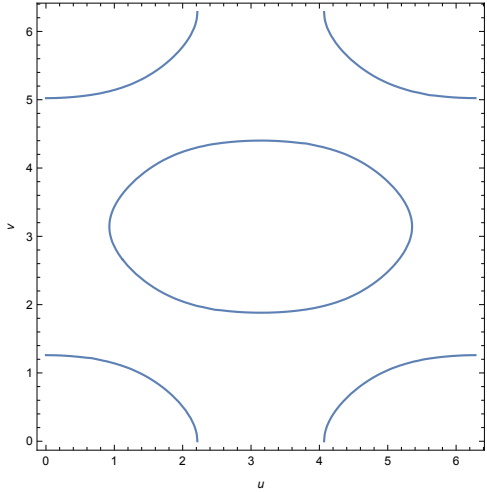}
\caption{Solution of \cref{eqn:brachiocyllosis} for $u,v \in [0,2\pi]$.}   \label{fig:Beethovenish}
\end{figure}

	\end{example}

	For $G \cong \Sp(4,3) \injto \UG(4)$, we find a zero on the submanifold $\UG(2)\times \UG(2)$. The size of the orbit is at most $\abs{G}^2/6 = 447897600$.

	\begin{example} \label{eg:pentahalide}
		Let $G \cong \Sp(4,3)$ be the matrix group in \cref{eg:zoon}, and let $f$ be the $G \times G$-invariant polynomial induced by the irreducible character $\chi_{(4,0,0,-4)}$ in \cref{eqn:northwestern}. 
		Then
		$U_0 = \begin{bmatrix}
			A & 0 \\
			0 & B 
		\end{bmatrix} $
		is a zero of $f$, where $A \approx \begin{bmatrix}
			-0.106632-0.973877 i & 0.0621677\, -0.190601 i \\
			 0.197341\, +0.0353545 i & -0.807683-0.554486 i
		\end{bmatrix}$ and $B \approx \begin{bmatrix}
			-0.596879-0.434093 i & -0.562033-0.373388 i \\
			0.372766\, -0.562445 i & -0.381284+0.631921 i 
		\end{bmatrix}$. 
		The error bound in \cref{alg:Pocket} is $\epsilon = 10^{-6}$.
		Hence the orbit $GU_0G$ is a unitary $4$-design on $\UG(4)$.
		The size of this orbit is $447897600$.
	\end{example}

	\begin{remark}
		The existence of exact $t$-designs are guaranteed by \cref{thm:Khazar}, which is different from finding approximate unitary $t$-designs. 
		\cref{alg:Pocket} can approximate such a unitary design with arbitrary precision if one has enough time and computational resources. 
		The time complexity of evaluating \cref{eqn:trifledom} is $O(tp(t)d^3 \abs{G}^2)$ where $p(n) \sim \frac{1}{4n \sqrt{3}} e^{\pi \sqrt{2n/3}}$, the partition function, is equal to the number of partitions of $n$. 
		The error $\epsilon$ is ideally halved after each iteration. 
		So it takes about $\log_2 10 \approx 3.3$ iterations to get one more significant digit. 
		For \cref{eg:pentahalide}, our program (written in Mathematica) ran on a PC equipped with Core i7-6700 CPU and 8GB RAM, and it took about half a day for each iteration.
	\end{remark}

\section{Discussion}

	It would be interesting to classify those unitary $t$-groups $G\subset U(d)$ that satisfy $(\chi^{t+1},\chi^{t+1})_G=(\chi^{t+1},\chi^{t+1})_{U(d)}+1=(t+1)! +1$, which is the condition of \cref{thm:Khazar}. 
	This should be certainly possible for $t\geq 2$ as such $G$ are among those already classified. 
	The problem would be interesting for $t=1$ as well. 
	We expect the existence of many such examples of unitary 2-designs by our method mentioned in this paper. 
	Such classification may be obtained by extending the method in Guralnick-Tiep \cite{MR2123127}, although actually doing so would not be trivial at all. 
	This would lead to explicit constructions of many families of explicit unitary $2$-designs. 
	We believe this is an independently interesting open problem from the viewpoint of finite group theory. 

	Concerning \cref{eg:cattlegate,eg:pentahalide}, it would be interesting to find what are the smallest sizes of unitary 3-designs, respectively 4-designs, that can be obtained by our method. 
	This may be done by discussing the possible submanifolds which contain the orbit. 
	If the function can achieve zero on a submanifold, we can still apply \cref{alg:Pocket}. 
	On the other hand it is not easy to show the non-existence of zeros on a submanifold. 

\section*{Acknowledgment}
The authors thank TGMRC (Three Gorges Mathematical Research Center) of China Three Gorges University in Yichang, Hubei, China, for supporting the visits of the authors to work on this research project in October 2018. 
The work is supported in part by NSFC Grant 11671258.
The work of M. N. is partly supported by
KAKENHI from JSPS Grant-in-Aid for Scientific Research
(KAKENHI Grant No. 17K05554). Y. Z. is supported by NSFC Grant No. 11801353 and China Postdoctoral Science Foundation No. 2018M632078.

\bibliographystyle{hplain}

\begin{thebibliography}{10}

\bibitem{ambainis2009nonmalleable}
Andris Ambainis, Jan Bouda, and Andreas Winter.
\newblock Nonmalleable encryption of quantum information.
\newblock {\em Journal of Mathematical Physics}, 50(4):042106, 2009.

\bibitem{MR740321}
Eiichi Bannai.
\newblock Spherical {$t$}-designs which are orbits of finite groups.
\newblock {\em J. Math. Soc. Japan}, 36(2):341--354, 1984.

\bibitem{MR2535394}
Eiichi Bannai and Etsuko Bannai.
\newblock A survey on spherical designs and algebraic combinatorics on spheres.
\newblock {\em European J. Combin.}, 30(6):1392--1425, 2009.

\bibitem{MR3594369}
Eiichi Bannai, Etsuko Bannai, Hajime Tanaka, and Yan Zhu.
\newblock Design theory from the viewpoint of algebraic combinatorics.
\newblock {\em Graphs Combin.}, 33(1):1--41, 2017.

\bibitem{1810.02507}
Eiichi Bannai, Gabriel Navarro, Noelia Rizo, and Pham~Huu Tiep.
\newblock Unitary t-groups, 2018, arXiv:1810.02507.

\bibitem{MR3535891}
Fernando G. S.~L. Brand\~{a}o, Aram~W. Harrow, and Micha\l Horodecki.
\newblock Local random quantum circuits are approximate polynomial-designs.
\newblock {\em Comm. Math. Phys.}, 346(2):397--434, 2016.

\bibitem{MR2062813}
Daniel Bump.
\newblock {\em Lie groups}, volume 225 of {\em Graduate Texts in Mathematics}.
\newblock Springer-Verlag, New York, 2004.

\bibitem{Cleve:2016:NCE:3179473.3179474}
Richard Cleve, Debbie Leung, Li~Liu, and Chunhao Wang.
\newblock Near-linear constructions of exact unitary 2-designs.
\newblock {\em Quantum Info. Comput.}, 16(9-10):721--756, July 2016.

\bibitem{MR827219}
J.~H. Conway, R.~T. Curtis, S.~P. Norton, R.~A. Parker, and R.~A. Wilson.
\newblock {\em Atlas of finite groups}.
\newblock Oxford University Press, Eynsham, 1985.
\newblock Maximal subgroups and ordinary characters for simple groups, With
  computational assistance from J. G. Thackray.

\bibitem{dankert2009exact}
Christoph Dankert, Richard Cleve, Joseph Emerson, and Etera Livine.
\newblock Exact and approximate unitary 2-designs and their application to
  fidelity estimation.
\newblock {\em Physical Review A}, 80(1):012304, 2009.

\bibitem{MR0384310}
P.~Delsarte.
\newblock An algebraic approach to the association schemes of coding theory.
\newblock {\em Philips Res. Rep. Suppl.}, (10):vi+97, 1973.

\bibitem{MR0485471}
P.~Delsarte, J.~M. Goethals, and J.~J. Seidel.
\newblock Spherical codes and designs.
\newblock {\em Geometriae Dedicata}, 6(3):363--388, 1977.

\bibitem{985948}
D.~P. {DiVincenzo}, D.~W. {Leung}, and B.~M. {Terhal}.
\newblock Quantum data hiding.
\newblock {\em IEEE Transactions on Information Theory}, 48(3):580--598, March
  2002.

\bibitem{GAP4}
The GAP~Group.
\newblock {\em {GAP -- Groups, Algorithms, and Programming, Version 4.10.1}},
  2019.

\bibitem{MR2326329}
D.~Gross, K.~Audenaert, and J.~Eisert.
\newblock Evenly distributed unitaries: on the structure of unitary designs.
\newblock {\em J. Math. Phys.}, 48(5):052104, 22, 2007.

\bibitem{MR2123127}
Robert~M. Guralnick and Pham~Huu Tiep.
\newblock Decompositions of small tensor powers and {L}arsen's conjecture.
\newblock {\em Represent. Theory}, 9:138--208, 2005.

\bibitem{MR2551028}
Aram~W. Harrow and Richard~A. Low.
\newblock Efficient quantum tensor product expanders and {$k$}-designs.
\newblock In {\em Approximation, randomization, and combinatorial
  optimization}, volume 5687 of {\em Lecture Notes in Comput. Sci.}, pages
  548--561. Springer, Berlin, 2009.

\bibitem{harrow2009random}
Aram~W Harrow and Richard~A Low.
\newblock Random quantum circuits are approximate 2-designs.
\newblock {\em Communications in Mathematical Physics}, 291(1):257--302, 2009.

\bibitem{PhysRevA.72.032325}
A.~Hayashi, T.~Hashimoto, and M.~Horibe.
\newblock Reexamination of optimal quantum state estimation of pure states.
\newblock {\em Phys. Rev. A}, 72:032325, Sep 2005.

\bibitem{PhysRevA.77.012307}
E.~Knill, D.~Leibfried, R.~Reichle, J.~Britton, R.~B. Blakestad, J.~D. Jost,
  C.~Langer, R.~Ozeri, S.~Seidelin, and D.~J. Wineland.
\newblock Randomized benchmarking of quantum gates.
\newblock {\em Phys. Rev. A}, 77:012307, Jan 2008.

\bibitem{1510.02767}
Richard Kueng and David Gross.
\newblock Qubit stabilizer states are complex projective 3-designs, 2015,
  arXiv:1510.02767.

\bibitem{2016APS..MARB44006L}
Y.-K. {Liu} and S.~{Kimmel}.
\newblock {Quantum Compressed Sensing Using 2-Designs}.
\newblock In {\em APS Meeting Abstracts}, page B44.006, 2016.

\bibitem{PhysRevLett.106.180504}
Easwar Magesan, J.~M. Gambetta, and Joseph Emerson.
\newblock Scalable and robust randomized benchmarking of quantum processes.
\newblock {\em Phys. Rev. Lett.}, 106:180504, May 2011.

\bibitem{PhysRevA.85.042311}
Easwar Magesan, Jay~M. Gambetta, and Joseph Emerson.
\newblock Characterizing quantum gates via randomized benchmarking.
\newblock {\em Phys. Rev. A}, 85:042311, Apr 2012.

\bibitem{Nakata}
Yoshifumi Nakata.
\newblock Unitary designs in quantum information science.
\newblock {\em Proceedings of the 35th Algebraic Combinatorics Symposium held
  in Hiroshima}, pages 89--100, Jun 2018.
\newblock {\href{https://hnozaki.jimdo.com/proceedings-symp-alg-comb/no-35/}{https://hnozaki.jimdo.com/proceedings-symp-alg-comb/no-35/}}
%{ https://hnozaki.jimdo.com/proceedings-symp-alg-comb/no-35/}.

\bibitem{PhysRevX.7.021006}
Yoshifumi Nakata, Christoph Hirche, Masato Koashi, and Andreas Winter.
\newblock Efficient quantum pseudorandomness with nearly time-independent
  hamiltonian dynamics.
\newblock {\em Phys. Rev. X}, 7:021006, Apr 2017.

\bibitem{nakata2017decoupling}
Yoshifumi Nakata, Christoph Hirche, Ciara Morgan, and Andreas Winter.
\newblock Decoupling with random diagonal unitaries.
\newblock {\em Quantum}, 1:18, 2017.

\bibitem{MR2529619}
Aidan Roy and A.~J. Scott.
\newblock Unitary designs and codes.
\newblock {\em Des. Codes Cryptogr.}, 53(1):13--31, 2009.

\bibitem{MR2433437}
A.~J. Scott.
\newblock Optimizing quantum process tomography with unitary 2-designs.
\newblock {\em J. Phys. A}, 41(5):055308, 26, 2008.

\bibitem{scott2006tight}
Andrew~J Scott.
\newblock Tight informationally complete quantum measurements.
\newblock {\em Journal of Physics A: Mathematical and General}, 39(43):13507,
  2006.

\bibitem{MR1740388}
Raymond S\'{e}roul.
\newblock {\em Programming for mathematicians}.
\newblock Universitext. Springer-Verlag, Berlin, 2000.
\newblock Translated from the 1995 French original by Donal O'Shea.

\bibitem{MR744857}
P.~D. Seymour and Thomas Zaslavsky.
\newblock Averaging sets: a generalization of mean values and spherical
  designs.
\newblock {\em Adv. in Math.}, 52(3):213--240, 1984.

\bibitem{MR0059914}
G.~C. Shephard and J.~A. Todd.
\newblock Finite unitary reflection groups.
\newblock {\em Canadian J. Math.}, 6:274--304, 1954.

\bibitem{szehr2013decoupling}
Oleg Szehr, Fr{\'e}d{\'e}ric Dupuis, Marco Tomamichel, and Renato Renner.
\newblock Decoupling with unitary approximate two-designs.
\newblock {\em New Journal of Physics}, 15(5):053022, 2013.

\bibitem{wallman2014randomized}
Joel~J Wallman and Steven~T Flammia.
\newblock Randomized benchmarking with confidence.
\newblock {\em New Journal of Physics}, 16(10):103032, 2014.

\bibitem{Webb:2016:CGF:3179439.3179447}
Zak Webb.
\newblock The clifford group forms a unitary 3-design.
\newblock {\em Quantum Info. Comput.}, 16(15-16):1379--1400, November 2016.

\bibitem{1711.08098}
Linxi Zhang, Chuanghua Zhu, and Changxing Pei.
\newblock Randomized benchmarking using unitary t-design for average fidelity
  estimation of practical quantum circuit, 2017, arXiv:1711.08098.

\bibitem{PhysRevA.90.012115}
Huangjun Zhu.
\newblock Quantum state estimation with informationally overcomplete
  measurements.
\newblock {\em Phys. Rev. A}, 90:012115, Jul 2014.

\bibitem{Zhu_2017}
Huangjun Zhu.
\newblock Multiqubit clifford groups are unitary 3-designs.
\newblock {\em Physical Review A}, 96(6), dec 2017.

\bibitem{PhysRevA.84.022327}
Huangjun Zhu and Berthold-Georg Englert.
\newblock Quantum state tomography with fully symmetric measurements and
  product measurements.
\newblock {\em Phys. Rev. A}, 84:022327, Aug 2011.

\bibitem{1609.08172}
Huangjun Zhu, Richard Kueng, Markus Grassl, and David Gross.
\newblock The clifford group fails gracefully to be a unitary 4-design, 2016,
  arXiv:1609.08172.

\end{thebibliography}

\end{document}